\documentclass[reprint,
 amsmath,amssymb,
 aps,
pra,
]{revtex4-2}
\usepackage{amsmath}
\usepackage{amssymb}
\usepackage{braket}
\usepackage{complexity}
\usepackage{mathtools}
\usepackage{amsmath,mleftright}
\usepackage{xparse}
\usepackage{amsthm}
\newtheorem{theorem}{Theorem}
\newtheorem{lemma}[theorem]{Lemma}
\usepackage{caption}
\usepackage[dvipsnames]{xcolor}

\usepackage{hyperref}
\usepackage{subcaption}
\DeclarePairedDelimiter\ceil{\lceil}{\rceil}
\DeclarePairedDelimiter\floor{\lfloor}{\rfloor}
\usepackage{graphicx}
\usepackage{dcolumn}
\usepackage{bm}

\newcommand\fig[1]{Fig.\ref{fig:#1}}

\usepackage{ulem}

\NewDocumentCommand{\evalat}{sO{\big}mm}{%
  \IfBooleanTF{#1}
   {\mleft. #3 \mright|_{#4}}
   {#3#2|_{#4}}%
}

\begin{document}

\preprint{APS/123-QED}

\title{Sample efficient graph classification using binary Gaussian boson sampling}

\author{Amanuel Anteneh}
\email{asa2rc@virginia.edu}
\affiliation{Department of Computer Science, University of Virginia, Charlottesville, Virginia 22903, USA}
\author{Olivier Pfister}%
\email{olivier.pfister@gmail.com}
\affiliation{Department of Physics, University of Virginia, Charlottesville, Virginia 22903, USA}

\date{\today}

\begin{abstract}
We present a variation of a quantum algorithm for the machine learning task of classification with graph-structured data.  The algorithm implements a feature extraction strategy that is based on Gaussian boson sampling (GBS) a near term model of quantum computing. However, unlike the currently proposed algorithms for this problem, our GBS setup only requires binary (light/no light) detectors, as opposed to photon-number-resolving detectors. Binary detectors are technologically simpler and can operate near room temperature, making our algorithm much less complex and costly to implement physically. We also investigate the connection between graph theory and the Torontonian matrix function which characterizes the probabilities of binary GBS detection events.
\end{abstract}

    \maketitle


\section{\label{sec:level1}Introduction}
Graphs are one of the most versatile data structures used in computing, and developing machine learning methods for working with graph-structured data has been a growing sub-field of machine learning research. Graph classification, in particular, has useful applications in fields such as bioinformatics, network science and computer vision as many of the objects studied in these fields can easily be represented as graphs.
However, using graph-structured data with machine learning models is not a straightforward task. This is because one of the most common ways of representing a graph for computational applications, i.e., as an adjacency matrix, cannot be easily used as an input to machine learning classifiers which primarily take vector-valued data as their inputs.
Therefore, a common way of working with graph-structured data is by defining a feature map $\phi$ that maps a graph $G$ to a vector in a Hilbert space called a feature space. From there a function $\kappa$, called a kernel, is defined that measures the similarity of two graphs in the feature space. An example of a feature map from $\mathbb{R}^2 \rightarrow \mathbb{R}^3$ is shown in Fig. \ref{fig:space}.

Kernel methods refer to machine learning algorithms that learn by comparing pairs of data points using this similarity measure. In our context we have a set of graphs $\mathbb{G}$ and we call a kernel $\kappa$ a \textit{graph kernel} if it is a function of the form $\kappa: \mathbb{G} \times \mathbb{G} \rightarrow \mathbb{R}$ \cite{nikolentzos2021graph,kriege2020survey}. The most common example of a kernel function is the feature space's inner product $\kappa(x,x')=\langle \phi(x),\phi(x') \rangle$. The goal of such methods is to construct mappings to feature vectors whose entries (the features) relate to relevant information about the graphs. Using a Gaussian boson sampling (GBS) device to construct graph kernels was an idea first proposed by Schuld et al.\ in Ref.~\citenum{schuld2020measuring}. 

\begin{figure}
\centering
\includegraphics[width=.8\columnwidth]{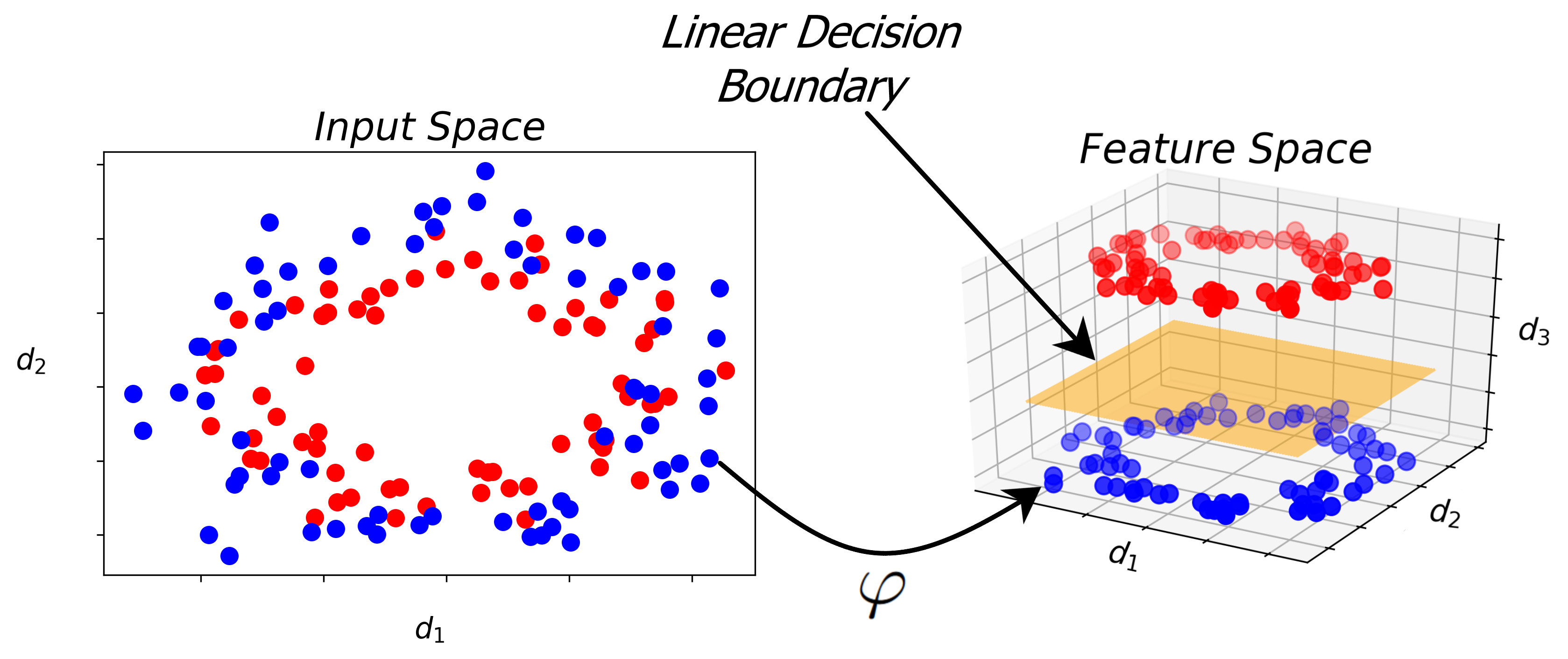}
\caption{In the original input space $\mathbb{R}^2$ the data points, which belong either to the class `red' or `blue', are not separable by a linear function (the decision boundary) but after mapping the points to feature vectors in a higher dimensional space $\mathbb{R}^3$ a linear function is able to separate the two classes. This linear decision boundary can be calculated by supervised machine learning models such as a support vector machine. In our case the input space is the set of all undirected graphs which we denote as $\mathbb{G}$.}
\label{fig:space}
\end{figure}

Boson sampling was first proposed by Aaronson and Arkhipov ~\cite{Aaronson2010} as a task---generating the photon-counting outcomes of the ``quantum Galton board'' constituted by an $M\times M$ optical interferometer fed with single photons into some of its input ports---that is strongly believed to be intractable to classical computers. The reason for this intractability is that calculating the probability distribution for generating random outcomes using Monte Carlo simulations requires calculating the permanent of an $M\times M$ matrix. Calculating the permanent of a general matrix is known to be $\#\P$-complete \cite{valiant1979complexity} which is a class of problems comparable to the class of $\NP$-complete problems in their difficulty. Gaussian boson sampling~\cite{Hamilton2017} is a variant of boson sampling in which the single-photon inputs are replaced with single-mode squeezed states, as produced, for example, by two-photon-emitting optical parametric amplifiers~\cite{Bachor2019}. The GBS probability distribution is governed by the Hafnian of an $M\times M$ matrix. Calculating the Hafnian of a general square matrix can be reduced to the task of calculating permanents therefore calculating the Hafnian is also $\#\P$-complete. In both cases, a quantum machine implementing boson sampling or GBS can easily sample from these hard-to-calculate probability distributions, just because they are ``wired-in,'' and this constitutes the ``quantum advantage'' that was recently demonstrated in optical experiments~\cite{Zhong2020,Madsen2022}. Note also that the initial ``quantum supremacy'' result obtained by Google on a superconducting qubit array~\cite{Arute2019} was a quantum (circuit) sampling result as well.

Beyond these necessary initial steps of demonstrating that quantum hardware can indeed reach regions inaccessible to classical hardware, a subsequent question is that of the utility of a sampling task. Whereas the usefulness of sampling in and of itself is far from established, we know that the histograms produced by statistically significant sampling constitute empirical probability distributions that tend toward the true, classically intractable probability distributions for sample numbers linear in the number of possible outcomes~\cite{Weissman2003}. The problem is that this very number of possible outcomes  grows exponentially with $M$ in a $M$-qubit quantum circuit in general~\footnote{Note that this is not related to the number of possible output quantum states, which scales with the number of parameters governing the quantum evolution, e.g.\ parameters of a simulated Hamiltonian. Obviously, no quantum advantage can be obtained for $M$-qubit Hamiltonians that have $\mathcal O(2^M)$ parameters but all classically intractable $M$-qubit Hamiltonians of physical interest are local and have parameter numbers polynomial in $M$~\cite{Lloyd1996}, which validates Feynman's proposed advantage for quantum simulation~\cite{Feynman1982}. An $M\times M$ optical interferometer has $M^2$ parameters, for example. However, even though any useful quantum computer will explore but a $\mathcal O(M^k)$-dimensional region of an $\mathcal O(2^M)$-dimensional Hilbert space, the number of measurement outcomes will still scale like $\mathcal O(2^M)$ a priori, simply because we do not know the adequate measurement basis that best contains the $\mathcal O(M^k)$ output states. This is the well known exponential overhead of quantum state tomography.}, and exponentially or super-exponentially with $M$ in an $M$-optical-mode boson or Gaussian boson sampler, which dispels any notion of quantum advantage for calculating the corresponding quantum probability distributions.

One direction that has been explored out of this conundrum is the binning of GBS measurements results into outcome classes whose cardinality scales favorably (e.g.\ polynomially) with the problem size (the GBS mode number). The immediate downside of such an approach is loss of information it entails, which impacts usefulness. However, graph classification using feature vectors and coarse-graining might provide advantageous GBS applications. This was first pointed out by Schuld et al.~\cite{schuld2020measuring}. 

In this paper, we show that a technologically simpler version of GBS, which we term binary GBS, can achieve comparable or better performance. The paper is structured as follows. In Sec.\ref{sec:level1} we give broad reminders about GBS and graph theory (with details in Appendix~\ref{app0}) and 
the current GBS graph kernel from Ref.~\citenum{schuld2020measuring}. We then present our graph kernel in Sec.\ref{sec:level2} along with results from numerical experiments and analyses of its complexity, features and advantages.

\begin{figure}
\centering
\includegraphics[scale=0.21]{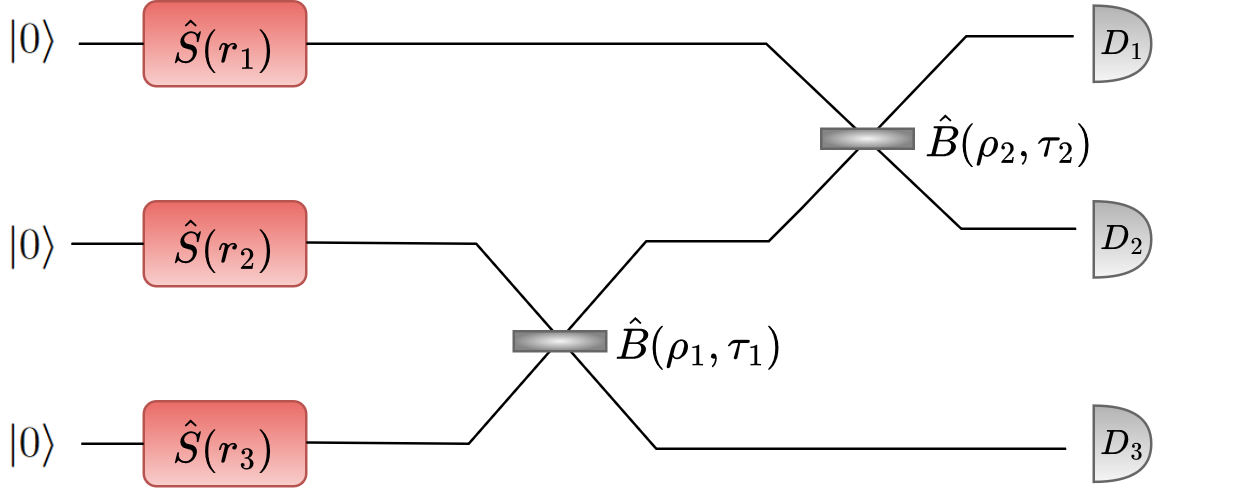}
\caption{Example of a 3-mode Gaussian boson sampler. Mode $i\in\{1,2,3\}$ starts in the vacuum state $\ket{0}$, is then squeezed by $\hat{S}(r_i)$ and passes through the network of two beamsplitters (the interferometer) before the number of photons in each mode is measured by the detectors $D_{i\in\{1,2,3\}}$.}
\label{fig:gbs}
\end{figure}

\section{\label{sec:level1} Reminders about Gaussian Boson Sampling (GBS) and graph theory}

\subsection{\label{sec:level2} Gaussian Boson Sampling}

As mentioned above, an $M$-mode GBS devise comprises $M$ single-mode-squeezing (SMS) inputs, an $M\times M$ optical interferometer, and $M$ photon-number-resolving (PNR) detectors, see \fig{gbs} for an example. The latter have come of age in superconducting devices such as transition edge sensors~\cite{Lita2008} and superconducting nanowire single-photon detectors~\cite{Cahall2017}. Both the former and the latter have recently been used to make PNR measurements of as many as 100 photons~\cite{Eaton2022_100,Cheng2022_100}. 

An $M$-mode Gaussian boson sampler prepares a Gaussian (Wigner function) quantum state by the $M$ squeezers and the interferometer. 
The squeezers output squeezed light into the interferometer and the photons are then passed through the interferometer after which the $M$ detectors detect what modes the photons end up in resulting in a detection event. We denote a detection event as $\textbf{n} = (n_1,...,n_M)$, where $n_i$ is the photon count in the $i$th mode and the total number of photons is $n=\sum_{i=1}^{M}n_i$. 

We now consider binary detectors, such as single-photon avalanche photodiodes, which are single-photon sensitive but aren't PNR and give the same signal however many photons were absorbed. In this case, we have $n_i \in \{0,1\}$ where $n_i=0$ indicated zero photons were detected in that mode and $n_i=1$ indicates that at least one photon was detected. When using binary detectors we no longer know the total photon number $n$ so we use $N$ to denote the number of detectors that detect photons leading to $\sum_{i=1}^{M}n_i = N \leq M$. 

An $M$-mode Gaussian state is fully described by a covariance matrix $\boldsymbol{\Sigma} \in \mathbb{R}^{2M \times 2M}$ and a displacement vector $\textbf{d} \in \mathbb{R}^{2M}$ \cite{weedbrook2012gaussian}.

\subsection{\label{sec:level2} Graph theory}

In this paper we define a graph $G=(V,E)$ as a set of vertices $V=\{v_1,v_2,...\}$ and a set of edges $E=\{(v_1,v_1),(v_1,v_2),... (v_i,v_j), ... \}$ that connect vertices if the edge value is not zero. A graph can be unweighted, with all nonzero edge weights equal to 1, or weighted, for example with real edge weights in GBS. For undirected graphs, which is what we will exclusively work with in this paper, $(v_i,v_j) = (v_j,v_i)$, $\forall i,j$. The size of a graph is equal to the cardinality $|V|$ of its vertex set.  
The degree of a vertex $v$ is the number of edges that are connected to it. The maximum degree of a graph is the largest degree of a vertex in its vertex set.

Graphs can be represented in a number of ways such as a diagram, Fig.\ref{fig:sub1}, or a more computationally useful way as an adjacency matrix, Fig.\ref{fig:sub2}. 
The adjacency matrix of an undirected graph $G$ with $|V|$ vertices is a $|V| \times |V|$ symmetric matrix $A$ with entries $a_{ij}$ where $a_{ij}$ is the weight of the edge connecting vertices $i$ and $j$. 
\begin{figure}
\centering
\begin{subfigure}{.25\textwidth}
  \centering
  \includegraphics[scale=0.5]{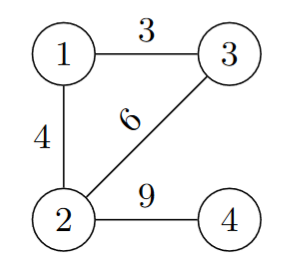}
  \caption{Undirected weighted 4-vertex graph with 4 edges}
  \label{fig:sub1}
\end{subfigure}%
\begin{subfigure}{.25\textwidth}
  \centering
 \includegraphics[scale=0.6]{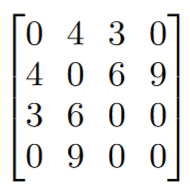}
  \caption{Adjacency matrix of graph}
  \label{fig:sub2}
\end{subfigure}
\caption{4-vertex graph and its corresponding adjacency matrix}
\label{fig:adjMat}
\end{figure}

{A substantial amount of work has been done on the connection between graph theory and Gaussian boson sampling with PNR detectors~\cite{bradler2018gaussian, bradler2021graph, arrazola2018using}. In Appendix~\ref{app0}, we recall some details of this work, namely how a given graph adjacency matrix can be encoded in a GBS experiment.}

\subsection{\label{sec:level3} Sample complexity of GBS}

\subsubsection{The problem with using GBS beyond sampling}

The sample complexity of a machine learning algorithm refers to the number of samples or amount of data required to learn some target function. In the case of GBS applications it refers to the number of samples we need to generate from the GBS device to learn or approximate a probability distribution over some set of the photon detection events.
This complexity type is extremely important to examine for any applications of GBS as it could potentially render certain applications of GBS intractable for larger problem sizes. 

For example it was shown that the GBS device utilizing PNR detectors can encode the graph isomorphism problem \cite{bradler2021graph}. 
This is done by encoding two graphs into two GBS devices and sampling each $S$ times. The $S$ samples could then be used, in principle, to reconstruct the probability distribution over all possible detection events $\textbf{n}$ for a given $M$ and $n$. However, this cannot be done efficiently enough to provide a quantum advantage. Indeed, we know from Refs.~\citenum{Weissman2003, canonne2020short} that reconstructing a probability distribution $D$ over a discrete finite set $\Omega$ of cardinality $|\Omega|$ from an empirical distribution $\hat{D}$ constructed from samples from $D$ we require
\begin{equation}\label{sampleEq}
S = \ceil*{ \frac{2(\ln(2)|\Omega| + \ln(\frac{1}{\delta}))}{\epsilon^2} }
\end{equation}
samples to guarantee that
\begin{equation}
    p(||D-\hat{D}||_1 \geq \epsilon) \leq \delta 
\end{equation}
where $||D-\hat{D}||_1$ denotes the $L_1$ distance between $D$ and $\hat{D}$.
In other words we require 
\begin{equation}
    \mathcal{O}\left( \frac{|\Omega|+\ln(\frac{1}{\delta})}{\epsilon^2}\right)
\end{equation} samples to ensure with probability at most $\delta$ that the sum of the absolute values of the errors on the empirical probability distribution is $\epsilon$ or greater.

This means the number of samples we need to approximate a probability distribution scales linearly with the number of elements in it's sample space i.e. the number of outcomes. In the case where $D$ is the probability distribution over the set of all possible PNR detection events the number of such events for a given number of modes $M$ and maximum number of photons $n$ is 
\begin{equation}
    |\Omega| = {n + M - 1 \choose M-1} = \frac{(n+M-1)!}{n!(M-1)!}
\end{equation}
which in number theory is also known as the formula for the number of weak compositions of an integer $n$ into $M$ parts. 
As shown in appendices \ref{app1} and \ref{app2} under the assumption that the number of modes scales quadratically with the number of photons, $M \in \mathcal{O} (n^2)$, this quantity grows super exponentially with $M$ and in general scales as $\mathcal{O}((n+M-1)^{M-1})$ meaning that as the size of the graphs increase, and therefore as the number of modes of the GBS device increase, we require an exponential number of samples to ensure the algorithm can give us the correct result within a certain probability. Therefore while the algorithm may in principle be able to decide graph isomorphism, it is sample inefficient to an exponential degree making it intractable to implement even with a fault tolerant quantum computer.

\subsubsection{Coarse graining of sample distributions}

However a method was suggested in \cite{bradler2021graph} to coarse-grain the probability distribution by combining outcomes into groups called orbits. Coarse-graining in this sense means to construct a new probability distribution over the set of these groups, the cardinality of which is less than the original set of all possible detection events. An orbit $\textrm{O}_{\textbf{n}}$ consists of a detection event $\textbf{n}$ and all of its permutations. For example the orbit that contains the detection event $\textbf{n} = (1,2,2)$ also contains the detection events $(2,1,2)$ and $(2,2,1)$. The number of orbits for a 4-mode GBS device is equal to the number of ways one can write $n_1+n_2+n_3+n_4=n$, where the order of the summands does not matter. This is called the number of partitions of the integer $n$ into $M$ parts and from the number theory literature \cite{orucc2016number} it is known to behave asymptotically as 
\begin{equation}
     |\Omega| \approx \frac{e^{\pi \sqrt{\frac{2(n-M)}{3}}}}{4\sqrt{3}(n-M)}, M \leq n \leq 2M.
\end{equation}
If we assume the number of photons grows linearly with the number of modes, $n \in \Theta(M) \rightarrow n = 2M$, we have the following asymptotic bound on the number of orbits
\begin{equation}
    \frac{1}{4\sqrt{3}M}e^{\pi \sqrt{\frac{2M}{3}}}  \in \mathcal{O}(\frac{e^{\pi \sqrt{\frac{2M}{3}}}}{M}).
\end{equation}
This means the number of orbits, which is now the number of outcomes $|\Omega|$ from Eq. \ref{sampleEq}, grows like $M^{-1}e^{\pi\sqrt{2M/3}}$ meaning we would have a sample complexity of 
\begin{equation}
\mathcal{O}\left( \frac{M^{-1}e^{\pi\sqrt{2M/3}}+\ln(\frac{1}{\delta})}{\epsilon^2}\right)    
\end{equation} which is subexponential but still intractable for large $M$.

\subsubsection{\label{sec:level3} Sample complexity of previously proposed GBS graph kernels}

The first GBS based graph kernel proposed in \cite{schuld2020measuring} maps a graph $G$ to feature vectors in a feature space $\phi: G \rightarrow \textbf{f} = (f_1,f_2,...,f_D) \in \mathbb{R}^D$. Where $f_i = p(\textrm{O}_\textbf{n}^i)$ is the probability of detecting a detection event from the orbit $\textrm{O}_{\textbf{n}}^i$. This kernel was shown to perform well against three of the four classical kernels we use as benchmarks in this paper. 
However a shortcoming of this method is that the sample complexity is $\mathcal{O}(M^{-1}e^{\pi\sqrt{2M/3}})$.

The second GBS kernel is of the form $\phi: G \rightarrow \textbf{f} = (f_1,f_2,...,f_D) \in \mathbb{R}^D$, with $f_i = p(\mathcal{M}_{n, \Delta_s}^i)$ where $p(\mathcal{M}_{n, \Delta_s}^i)$ is the probability of detecting a detection event that belongs to the ``meta-orbit'' $\mathcal{M}_{n, \Delta_s}^i$. 
A meta-orbit $\mathcal{M}_{n, \Delta_s}$ is uniquely defined by a total photon number $n$ and $\Delta_s$ 
which is defined as
\begin{equation}
    \Delta_s = \{\textbf{n} \textrm{ : } \sum_{i}n_i = n \; \land \; \forall i \textrm{ : } n_i \leq  s \; \}.
\end{equation} 
Therefore a meta-orbit consists of all detection events where total photon number is equal to $n$, where no detector counts more than $s$ photons. It is claimed that this strategy partitions the set of all PNR detection events into a polynomial number of subsets in $n$ \cite{bradler2019duality}.
\section{\label{sec:level2}The Algorithm}

\subsection{\label{GBS-Threshold-Theory} GBS with binary detectors and its relation to graph theory}

While the relationship between GBS with PNR detectors and graph theory has been thoroughly explored, there has been little exploration of how GBS with binary detectors fits into the picture. In this section we shed some light on the relationship between the two. As stated before when using binary detectors the detection outcomes are of the form $\textbf{n}_{\textrm{bin}}=(n_1,...,n_M)$ where $n_i \in \{0,1\} $ $\forall i$ and $n_i=1$ indicates the $i$th detector detected one or more photons. The probability of detecting a detection outcome with binary detectors is characterized by a matrix function called the Torontonian, to which the same arguments for classical intractability as for the Hafnian can be extended~\cite{quesada2018gaussian}. The probability of a given binary detection event $\textbf{n}_{\textrm{bin}}$ is given by 
\begin{equation}
    p(\textbf{n}_{\textrm{bin}}) = \frac{\textrm{Tor}(O_{\textbf{n}_{\textrm{bin}}})}{\sqrt{\textrm{det}(Q)}} = \frac{\textrm{Tor}(X\tilde{A}_{\textbf{n}_{\textrm{bin}}})}{\sqrt{\textrm{det}(Q)}}
\end{equation}
{
where 
\begin{align}
    \tilde{A}&=(A \oplus A), \\
    X &=  \begin{bmatrix}
    0 & \mathbb{I}\\
   \mathbb{I} & 0 
    \end{bmatrix}, \\
    Q& = (\mathbb{I}_{2M} - X\tilde{A})^{-1}, \\
    O& = \mathbb{I}-Q^{-1}
\end{align}
}
and Tor() is the Torontonian of a $2N\times 2N$ matrix $A$ defined as 
\begin{equation}
        \textrm{Tor}(A) = \sum_{Z \in P([N])} (-1)^{|Z|} \frac{1}{\sqrt{\textrm{det}(\mathbb{I}-A_Z)}}.
\end{equation}
Where $P([N])$ is the power set, the set of all possible subsets, of the set $[N]=\{1,2,...,N\}$.
The probability of a PNR detection event $\textbf{n}$ can be written in terms of the matrix $O$ as 
\begin{equation}
    p(\textbf{n}) =  \frac{1}{\sqrt{\textrm{det}(Q)}}\frac{\textrm{Haf}(\tilde{A}_{\textbf{n}})}{\textbf{n}!}  = \frac{1}{\sqrt{\textrm{det}(Q)}}\frac{\textrm{Haf}(XO_{\textbf{n}})}{\textbf{n}!}.
\end{equation}

The probability of a binary GBS detection event is simply the sum of all probabilities of the corresponding PNR detection events. 
A useful example to illustrate this is a 4-mode Gaussian boson sampler programmed according to some adjacency matrix $A$ of a graph $G$. Suppose we use binary detectors and measure the detection event $\textbf{n}_\textrm{bin} = (1,0,1,0)$. The corresponding detection events when using PNR detectors would be of the form $\textbf{n} = (n_1,0,n_3,0)$ where $n_1,n_3 > 0$  and $n_1+n_3$ is even. We will define $\mathcal{N}$ to be the set of all possible $4$-mode PNR detection events with 0's in the 2nd and 4th index, i.e. only the 2nd and 4th detectors detect no photons. From this we have 
\begin{equation}
\begin{split}
    p((1,0,1,0)) = \frac{\textrm{Tor}(X\tilde{A}_{(1,0,1,0)})}{\sqrt{\textrm{det}(Q)}}  \\ = \sum_{\textbf{n} \in \mathcal{N}}p(\textbf{n})  = \sum_{\textbf{n} \in \mathcal{N}} \frac{\textrm{Haf}^2(A_{\textbf{n}})}{\textbf{n}!\sqrt{\textrm{det}(Q)}}.
\end{split}
\end{equation}
This means the Torontonian of $X\tilde{A}$ is proportional to an infinite sum of Hafnians as there are an infinite number of integer lists of the form $(n_1,0,n_3,0)$ where $n_1,n_3 > 0$. In a real GBS experiment, however, the energy is finite and therefore the measured probabilities of these events would be equal to a finite version of this sum where all detection events with total photon number greater than some cutoff photon number vanish from the series.

In terms of graph theory this means the probability of detecting $\textbf{n}_\textrm{bin}= (1,0,1,0)$ is proportional to the sum of the squared Hafnians of all possible subgraphs of $G$ of unbounded and even size with their 2nd and 4th vertices removed. But again in practice the maximum size of the subgraphs will always be bounded by some maximum photon number for a real GBS experiment. More generally we have 
\begin{equation}
    p(\textbf{n}_\textrm{bin}) = \frac{\textrm{Tor}(O_{\textbf{n}_{\textrm{bin}}})}{\sqrt{\textrm{det}(Q)}} = \frac{\textrm{Tor}(X\tilde{A}_{\textbf{n}_{\textrm{bin}}})}{\sqrt{\textrm{det}(Q)}} = \sum_{\textbf{n} \in \mathcal{N}} \frac{\textrm{Haf}^2(A_{\textbf{n}})}{\textbf{n}!\sqrt{\textrm{det}(Q)}}
\end{equation}
where $\mathcal{N}$ is the set of all PNR events that correspond to the binary detection event $\textbf{n}_\textrm{bin}$ \footnote{The Torontonian was also shown to be a generating function for the Hafnian given by \begin{equation}
    \textrm{Haf}(A) = \frac{1}{M!}\frac{d^M}{dz^M}
    \evalat{\textrm{Tor}(zXA)}{z=0}
\end{equation}
where $A$ is a $2M\times 2M$ matrix.
}.

\subsection{Constructing the feature vectors}
Once the GBS device is programmed we generate $S$ samples from the device. For our algorithm we use binary detectors so each sample is a list of length $M$ with entries either 0 or 1. Once we have these samples we use them to construct the feature vector of which we have two definitions based on two coarse-graining strategies. 

The first is based on what we call the $\mu$ coarse-graining strategy where we group together detection events that contain exactly $i$ detector `clicks' or ones. For example the detection events $(1,0,0)$ and $(0,0,1)$ would be grouped together since they both contain exactly 1 detector click. These groups can also be thought of as `binary orbits' since they contain a detection event and all its permutations. This strategy partitions the set of all binary detection events into a linear number of disjoint subsets in $N$. Using this strategy we can define the feature map as $\phi:G\rightarrow\textbf{f} = (f_0,f_1,...,f_{N}) \in \mathbb{R}^{N}$. Where $N$ is the maximum number of detector clicks and $f_i=\frac{S_i}{S}$ with $S_i$ being the number of samples which contain exactly $i$ ones. Equivalently this is the probability of detecting an event where exactly $i$ detectors detect a photon. 

The second feature map is based on what we call the $\nu$ coarse-graining strategy. For a 5 mode boson sampler utilizing binary detectors with maximum click number 5 there are $|\Omega|=32$ possible detection outcomes. This coarse-graining strategy groups together detection events whose first 5 modes are one of these 32 outcomes. For example the detection event $\textbf{n}_{\textrm{bin}}=(0,1,0,0,1,0,1)$ belongs in the group associated with the detection event
$(0,1,0,0,1)$ since they are equal if one is only concerned with the first 5 modes. This strategy partitions the set of all detection events of $5$ or more modes into a constant number of subsets, i.e. 32. The feature map based on this strategy is defined as $\phi:G\rightarrow\textbf{f} = (f_{[0,0,0,0,0]},f_{[1,0,0,0,0]},...,f_{\textbf{n}}) \in \mathbb{R}^{32}$. Where $f_{\textbf{n}}$ is the probability of detecting an event where the first 5 modes correspond to one of the 32 possible detection outcomes. For example $f_{[1,0,0,0,0]}$ is the probability that the first detector detects photons and the following 4 detectors detect vacuum.

Once we construct the feature vector for each graph in the data set we input them to a machine learning classifier such as a support vector machine.

\begin{table*}
\caption{\label{tab:table1} Graph data set statistics after prepossessing. A more detailed description of these data sets can be
found in appendix B of Ref.~\citenum{schuld2020measuring}. }

\begin{tabular}{|c|c|c|c|c|}
\hline
 Data set &  \# of  graphs &  \# of classes &  avg. \# of vertices &  avg. \# of edges\\ \hline
 
AIDS & $1723$ & $2$ & $11.11$ & $11.29$  \\

BZR\_MD & $257$ & $2$ & $20.10$ & $197.69$ \\

COX2\_MD & $118$ & $2$ & $23.90$ & $274.40$  \\

ENZYMES & $204$ & $6$ & $18.56$ & $36.30$ \\

ER\_MD & $357$ & $2$ & $19.27$ & $185.15$ \\

FINGERPRINT & $1080$ & $3$ & $10.58$ & $9.10$ \\

IMDB-BINARY & $806$ & $2$ & $15.98$ & $63.32$ \\

MUTAG & $179$ & $2$ & $17.48$ & $19.23$ \\

NCI1 & $1853$ & $2$ & $19.77$ & $21.27$  \\

PROTEINS & $515$ & $2$ & $15.77$ & $29.37$ \\

PTC\_FM & $284$ & $2$ & $13.64$ & $13.99$ \\
\hline
\end{tabular}

\end{table*}

\subsection{\label{sec:level5} Complexity analysis}
In this section we discuss, in addition to the time and space complexity, the sample complexity of our algorithm.
\subsubsection{Sample Complexity}
Since the $n_i$'s for binary detection events can be either $0$ or $1$ we can think of the detection outcomes as binary strings of length $M$ with at most $M$ ones. 
The number of binary strings of length $M$ with exactly $i$ ones is ${M \choose i}$. So the number of possible binary detection events, the number of binary strings of length $M$ with at most $M$ ones, is given by 
\begin{equation}
  |\Omega| = \sum_{i=0}^{M} {M \choose i }.  
\end{equation}
We can show this function grows like $2^M$ using the binomial expansion
\begin{equation}
    2^M = (1+1)^M = \sum_{i=0}^{M} {M \choose i} 1^{M-i} 1^{i} = \sum_{i=0}^{M} {M \choose i}.
\end{equation}
Therefore we could not simply use the probability of the individual detection events as features without coarse-graining even when using binary detectors as we would still need a prohibitively large number of samples to approximate their probabilities to within a constant error. This was the reason for introducing the $\nu$ and $\mu$ coarse-graining strategies. 

Since the number of outcomes of the $\mu$ distribution scales linearly with $N$ which is $\leq M$ the sample complexity of approximating the $\mu$ coarse-grained probability distribution is 
\begin{equation}
\mathcal{O}\left( \frac{M+\ln(\frac{1}{\delta})}{\epsilon^2}\right)    
\end{equation}
which reduces to $\mathcal{O}(M)$ for constant $\epsilon$ and $\delta$.
The sample complexity of approximating the $\nu$ coarse-grained probability distribution is 
\begin{equation}
\mathcal{O}\left( \frac{32+\ln(\frac{1}{\delta})}{\epsilon^2}\right)    
\end{equation}
which reduces to $\mathcal{O}(1)$ for constant $\epsilon$ and $\delta$.
\subsubsection{Space Complexity}
The size of the $\nu$ feature vectors is constant with respect to the graph size so the space required is $\mathcal{O}(1)$ and for the $\mu$ feature vectors the size grows linearly with $N$ which is $\leq M$ so the space required is $\mathcal{O}(M)$.  However storing the adjacency matrix of the graphs requires $\mathcal{O}(M^2)$ space complexity. 
\subsubsection{Time Complexity}
The time complexity is determined by the most computationally time intensive step of the algorithm which is encoding the adjacency matrix into the GBS device. This is the case because the encoding process requires taking the Takagi decomposition of the matrix $A$ which for a $M \times M$ matrix has time complexity $\mathcal{O}(M^3)$ as it is a special case of the singular value decomposition
\cite{hahn2006routines}. However there do exist quantum algorithms for computing the singular value decomposition of a matrix with complexity that is polylogarithmic in the size of the matrix \cite{gu2019quantum}. In particular the quantum singular value estimation algorithm for a $m\times n$ matrix presented in \cite{kerenidis2016quantum} has complexity $\mathcal{O}(\textrm{polylog}(mn)/\epsilon)$ where $\epsilon$ is an additive error.

{
\section{\label{sec:level4} Numerical experiments}}
{
\subsection{Implementation details}}
We used \texttt{The Walrus} python library to classically sample from the GBS output distribution when running our experiments and the \texttt{GraKel} python library to fetch the data sets and simulate the classical graph kernels \cite{gupt2019walrus, siglidis2020grakel}. Classically sampling from a GBS output distribution is very time intensive even when using binary detectors so we choose to follow the choice made in \cite{schuld2020measuring} and discard graphs with greater than 25 and less than 6 vertices for each data set. Before sampling from the GBS device we have four parameters we can set: the maximum number of detector clicks allowed $N$, the average photon number $\bar{n}$, the displacement on each mode of the GBS device $d$ and lastly the number of samples generated by the GBS device $S$. 
We set $N=6, \bar{n}=5$ and $d=0$ for our results reported here leading to probability distribution of 32 outcomes using the $\nu$ coarse-graining strategy and $7$ outcomes using the $\mu$ coarse-graining strategy. Using Eq. \ref{sampleEq} with $\delta=0.01$ and $\epsilon=0.06$ we require about $S=15000$ samples for the $\nu$ feature vectors and about $S=6000$ samples for the $\mu$ feature vectors.

For the machine learning classifier we use a support vector machine with an RBF kernel $\kappa_{\textrm{rbf}}$. We obtain the accuracies in Table \ref{tab:table2} by running a double 10-fold cross-validation 10 times. The inner fold performs a grid search through the discrete set of values $[10^{-4}, 10^{-3},...,10^{2},10^{3}]$ on the $C$ hyper-parameter of the SVM which controls the penalty on misclassifications. 

{
\subsection{Numerical results from GBS simulation and subsequent classification}}

We tested our graph kernel on the same data sets used in \cite{schuld2020measuring}. We also ignored vertex labels, vertex attributes and edge attributes and converted all adjacency matrices to be unweighted.

Four classical graph kernels were used as a benchmark for our algorithms classification accuracy. The subgraph matching kernel (SM) with time complexity $\mathcal O(kM^{k+1})$ where $M$ is the number of vertices and $k$ the size of the subgraphs being considered \cite{kriege2012subgraph}, the graphlet sampling kernel (GS) with worst case time complexity $\mathcal O(M^k)$ which can be optimized to $\mathcal{O}(Md^{k-1})$ for graphs of bounded degree with the restriction that $k \in \{3,4,5\}$, where $k$ is the graphlet size and $d$ is the maximum degree of the graph \cite{shervashidze2009efficient}, the random walk kernel (RW) with time complexity $\mathcal O(M^3)$ \cite{vishwanathan2010graph} and the shortest path kernel (SP) with time complexity $\mathcal{O}(M^4)$ \cite{borgwardt2005shortest}. For the graphlet sampling kernel we set maximum graphlet size to $k=5$ and draw $5174$ samples, for the random walk kernel we use fast computation and a geometric kernel type with the decay factor set to $\lambda=10^{-3}$, for the subgraph matching kernel we set maximum subgraph size to $k=5$ and for the shortest path kernel we used the Floyd–Warshall algorithm to calculate shortest paths. 
{The accuracies of all four classical kernels and our kernel are shown in the top portion of Table \ref{tab:table2}. The values in the bottom portion of Table \ref{tab:table2} are the accuracies of the original GBS graph kernels and are taken from \cite{schuld2020measuring} where the features vectors were constructed with $n=6$. Some accuracies for the subgraph matching kernel are not reported due to it's $\mathcal{O}(M^6)$ time complexity for $k=5$ which required longer than 7 days of computation time for data sets with a high average number of vertices.}

{
We can see from the table that our kernel is very competitive with both the classical and PNR based GBS graph kernels and in fact achieves the highest accuracy on the ENZYMES dataset. The PNR based kernel only obtains a test accuracy significantly higher than random guessing ($\approx 16\%$) for ENZYMES only when displacement is applied. Our kernel can be seen as even more feasible in this regard since we do not require the extra operation of displacement to reach our level of accuracy. From Fig. \ref{fig:size} we see that for some datasets such as AIDS and MUTAG there is a strong size imbalance amongst graphs of different classes. For the AIDS dataset graphs belonging to class 1 are of much smaller size than those in class 0 while for MUTAG the converse is true. This size imbalance also exists to a lesser extent for the ER\_MD, BZR\_MD and PTC\_FM datasets and for the FINGERPRINT dataset class 1 has a graph size distribution significantly different from the other two. Both our GBS kernel and the PNR based kernel perform well on these size imbalanced datasets which indicates that graph size is a property that GBS based kernels are sensitive to and that this sensitivity persists even when binary detectors are used.
}

\begin{table*}
\begin{ruledtabular}
\caption{\label{tab:table2}
{
Average test accuracies of the support vector machine with different data sets and graph kernels.  The values in bold in the upper/lower section of the table are the best accuracy obtained for that section. The values in parenthesis are the standard deviation across the 10 repeats of double cross validation. GS, RW, SM and SP refer to the graphlet sampling, random walk, subgraph matching and shortest path kernels respectively. $\protect \textrm{GBS}^\textrm{bin}_{\nu}$ and $\protect \textrm{GBS}^\textrm{bin}_{\mu}$ denotes our GBS kernel with binary detectors that use the $\nu$ and $\mu$ coarse-graining strategies to construct the feature vectors respectively. $\protect \textrm{GBS}^\textrm{bin+}_{\nu}$ denotes that the feature associated with detecting vacuum $[0,0,0,0,0]$ in the first 5 modes was dropped from all feature vectors. $\protect \textrm{GBS}^{\textrm{PNR}}$ and $\protect \textrm{GBS}^{\textrm{PNR+}}$ refer to the original GBS kernels with PNR detectors that use orbit and meta-orbit probabilities as features respectively with a displacement of $d$ on each mode. 
*Runtime $>$ 7 days}
}
\begin{tabular}{cccccccccccc}
 Data set & $\rm GBS^\textrm{bin+}_{\nu}$ & $\rm GBS^\textrm{bin}_{\nu}$ & $\rm GBS^\textrm{bin}_{\mu}$ & GS & RW & SM & SP\\ \hline
 
AIDS & $98.47(0.10)$ & $98.74(0.20)$ & $\textbf{99.53}(0.05)$ & $99.30(0.07)$ & $53.11(11.90)$ & $77.85(2.44)$ & $99.34(0.09)$ \\

BZR\_MD & $60.14(1.28)$ & $61.73(0.89)$ & $58.79(1.17)$ & $51.42(3.51)$ & $\textbf{64.54}(0.36)$ & time out* & $50.82(1.76)$ \\

COX2\_MD & $\textbf{51.62}(2.76)$ & $50.18(2.96)$ & $51.30(3.86)$ & $49.01(3.18)$ & $48.98(4.78)$ & time out* & $48.11(4.30)$\\

ENZYMES & $\textbf{48.10}(1.18)$ & $41.75(2.35)$ & $19.83(1.43)$ & $34.59(2.54)$ & $19.50(2.29)$ & $37.38(1.60)$ & $22.15(1.88)$\\

ER\_MD & $67.74(0.94)$ & $69.19(0.33)$ & $68.84(0.50)$ & $48.88(4.53)$ & $\textbf{70.32}(0.02)$ & time out* & $45.23(4.35)$\\

FINGERPRINT & $64.45(0.78)$ & $\textbf{65.53}(0.86)$ & $63.56(0.67)$ & $65.25(1.30)$ & $33.63(3.57)$ & $46.89(0.56)$ & $46.22(1.02)$\\

IMDB-BINARY & $60.69(0.84)$ & $61.35(0.98)$ & $67.34(0.38)$ & $\textbf{68.49}(0.63)$ & $67.78(0.38)$ & time out* & $65.50(0.27)$\\

MUTAG & $84.63(0.91)$ & $\textbf{85.94} (0.98)$ & $81.37(0.90)$ & $80.80(0.91)$ & $83.22(0.04)$ & $83.24(1.27)$ & $82.74(1.65)$\\

NCI1 & $\textbf{63.45}(0.57)$ & $56.99(1.69)$ & $59.09(1.02)$ & $50.34(3.22)$ & $50.96(3.58)$ & time out* & $53.40(2.25)$ \\

PROTEINS & $\textbf{65.95}(1.03)$ & $63.38(0.73)$ & $63.11(0.55)$ & $65.75(0.94)$ & $56.91(1.39)$ & $62.93(0.83)$ & $63.63(0.41)$\\

PTC\_FM & $52.63(3.95)$ & $57.47(2.72)$ & $59.17(1.58)$ & $\textbf{60.74}(1.48)$ & $50.95(3.68)$ & $56.36(2.66)$ & $55.38(4.04)$\\
\end{tabular}

\begin{tabular}{ccccc} 
     Data set & $\textrm{GBS}^{\textrm{PNR}}$ ($d=0$) & $\textrm{GBS}^{\textrm{PNR}}$ ($d=0.25$) & $\textrm{GBS}^{\textrm{PNR+}}$ ($d=0$) & $\textrm{GBS}^{\textrm{PNR+}}$ ($d=0.25$) \\  \hline
     AIDS & $99.60 (0.05)$ & $\textbf{99.62} (0.03)$ & $99.58(0.06)$ & $99.61 (0.05)$  \\
     BZR\_MD & 62.73(0.71) & 62.13(1.44) & 62.01(1.43) & \textbf{63.16}(2.11) \\
     COX2\_MD & 44.98(1.80) & 50.11(0.97) & 57.84(4.04) & \textbf{57.89}(2.62) \\
     ENZYMES & 22.29(1.60) & 28.01(1.83) & 25.72(2.60) & \textbf{40.42}(2.02) \\
     ER\_MD & 70.36(0.78) & 70.41(0.47) & 71.01(1.26) & \textbf{71.05}(0.83) \\
     FINGERPRINT & 65.42(0.49) & \textbf{65.85}(0.36) & 66.19(0.84) & 66.26(4.29) \\     
     IMDB-BINARY & 64.09(0.34) & \textbf{68.71}(0.59) & 68.14(0.71) & 67.60(0.75) \\
     MUTAG & \textbf{86.41}(0.33) & 85.58(0.59) & 85.64(0.78) & 84.46(0.44) \\
     NCI1 & \textbf{63.61}(0.00) & 62.79(0.00) & 63.59(0.17) & 63.11(0.93) \\
     PROTEINS & \textbf{66.88}(0.22) & 66.14(0.48) & 65.73(0.69) & 66.16(0.76) \\
     PTC\_FM & 53.84(0.96) & 52.45(1.78) & \textbf{59.14}(1.72) & 56.25(2.04) \\
\end{tabular}
\end{ruledtabular}
\end{table*}

\subsection{\label{sec:level6} Feature analysis}
\begin{figure*}
\centering
\includegraphics[scale=0.33]{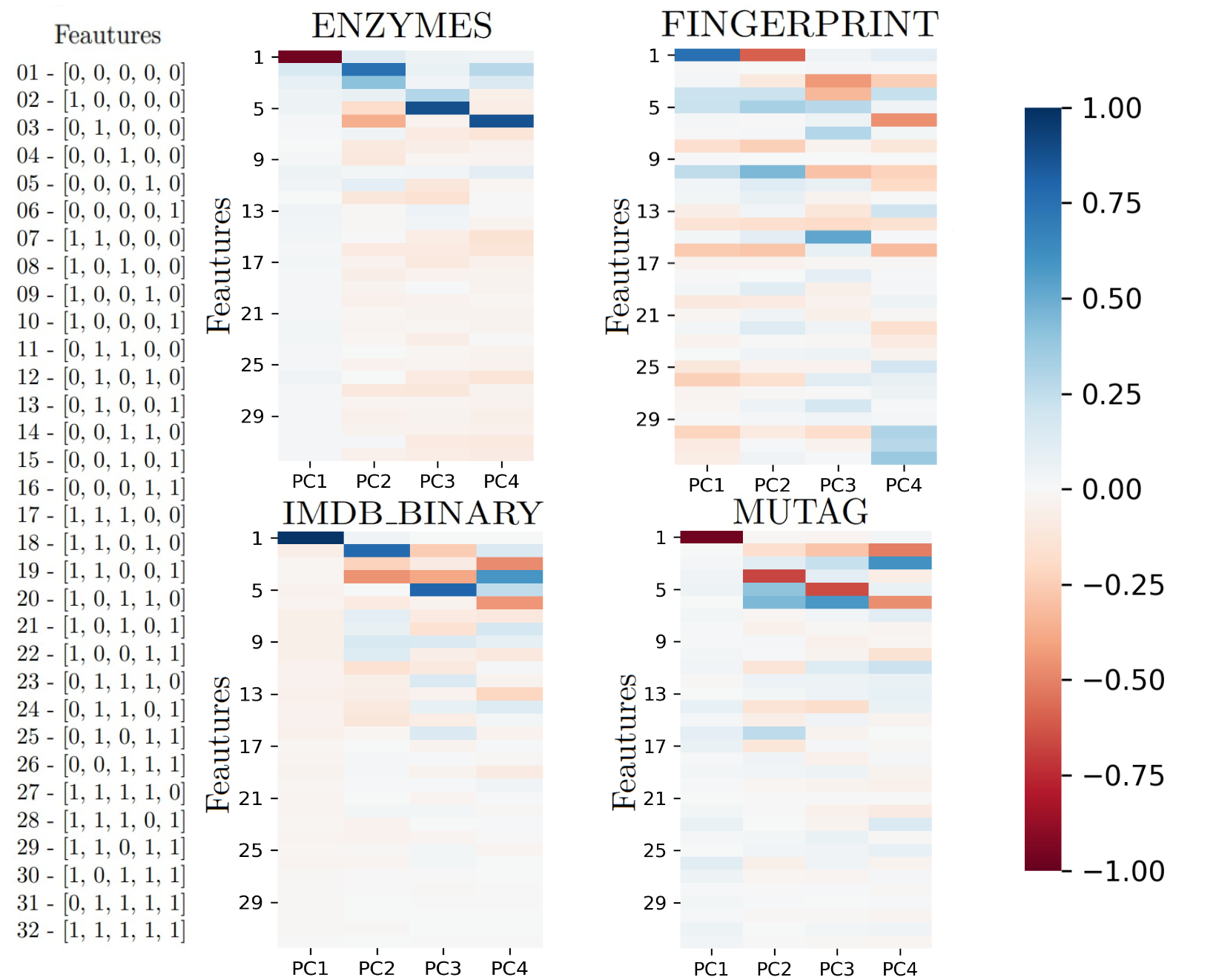}
\caption{Results of the principal component analysis (PCA) on the $\nu$ feature vector entries for the ENZYMES, MUTAG, IMDB\_BINARY and FINGERPRINT datasets. The heatmaps show the weight/coefficient associated with each feature with regard to the first four principal components.}
\label{fig:pca}
\end{figure*}

\begin{figure*}
\centering
\includegraphics[scale=0.2]{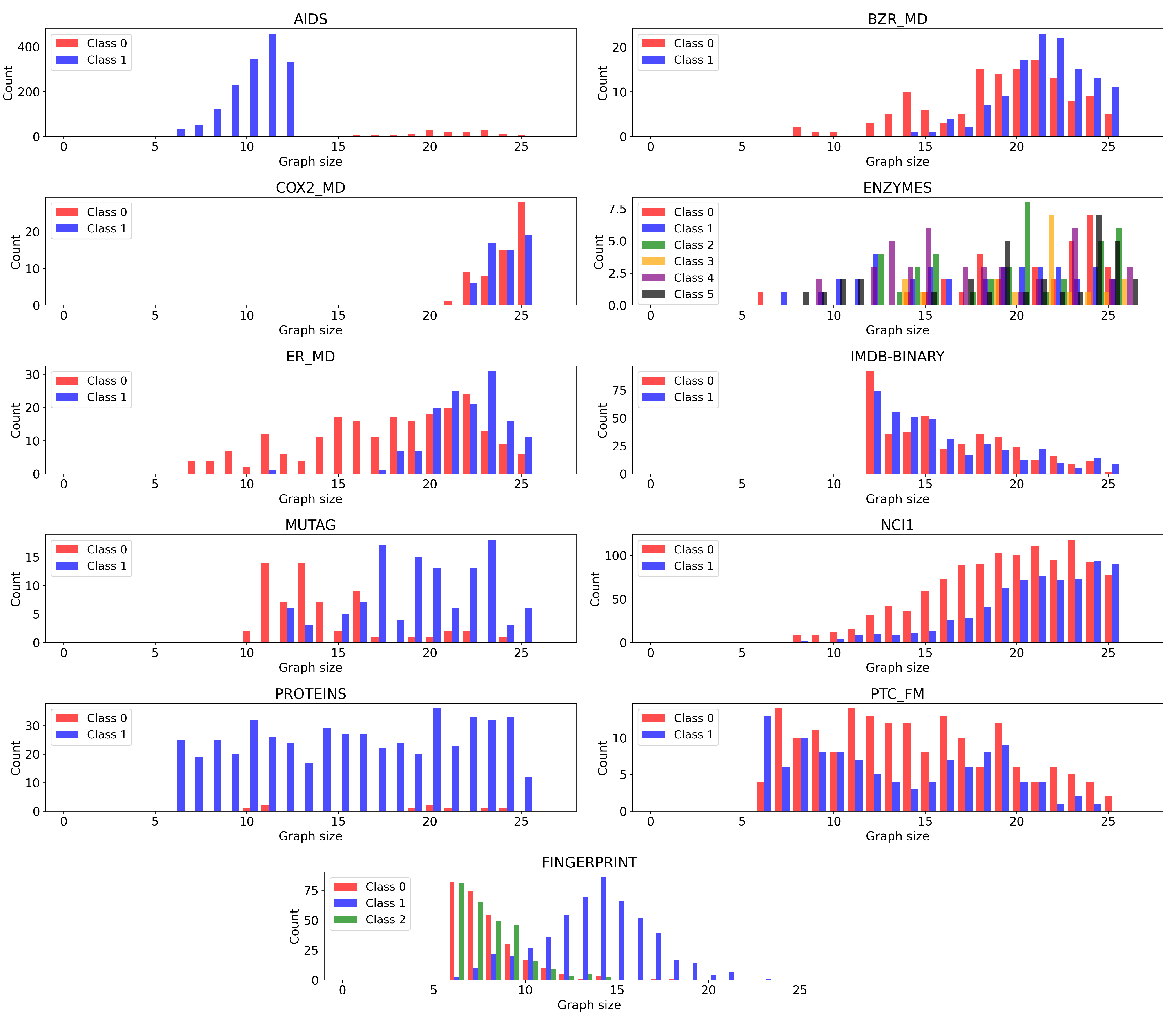}
\caption{Distribution of graph sizes according to class for each dataset. }
\label{fig:size}
\end{figure*}

Fig. \ref{fig:pca} shows the results of performing a principal component analysis on the feature vectors generated using the $\nu$ coarse-graining strategy for various datasets. The analysis shows that the feature associated with vacuum $[0,0,0,0,0]$ contributes by far the most in the support of the first principal component. The analysis also suggests that in some cases the first 10 or so features contribute the most to the support of all of the first four principal components but in other cases, such as with FINGERPRINT, most features contribute more or less equally.

\subsection{\label{sec:level7} Comparison to classical kernels}

Our graph kernel has a time complexity that is equivalent to the random walk kernel and better than the shortest path kernel by a factor of $M$ while outperforming both on most data sets. Furthermore the time complexity of our kernel is not exponential in the size of the subgraphs we are probing like the subgraph matching kernel.
The graphlet sampling kernel does have a more favorable complexity of $\mathcal{O}(Md^{k-1})$ for graphs with maximum degree $d$. However it's important to note that many real world graphs are what are called `scale-free networks' and from the network science literature \cite{barabasi2016} the maximum degree of these graphs grows polynomially with the graph size. Therefore it is possible that the maximum degree of these graphs grows linearly with the graph size i.g. $d \in \mathcal{O}(M)$ which would lead to a complexity of $\mathcal{O}(M^{k})$ for the graphlet sampling kernel. What is also interesting is that GBS kernels seems to provide more distinguishing power than some classical kernels for graphs with no vertex and edge labels like those used in our simulations. Take for example the ENZYMES dataset for which the binary GBS kernel achieves a classification accuracy of $\approx48\%$ while the shortest path kernel reaches about 23\%. If we instead choose to not ignore vertex labels we found the shortest path kernel gives a classification accuracy of about $50\%$. Since the GBS features are related to Hafnians this suggests that features related to the number of perfect matchings of a graph could be more useful for distinguishing graphs of different classes when one has no information about the attributes of the graph nodes.

\section{\label{sec:level4}Conclusion}
We proposed a variation of an algorithm for the machine learning task of classification with graph-structured data that uses a Gaussian boson sampler utilizing only binary detectors. We show that our algorithm outperforms four classical graph kernels for the task of graph classification on many data sets. This is most evident with regard to the ENZYMES data set where the $\nu$ feature map outperforms all methods. The feature corresponding to detecting vacuum in the first 5 modes plays a particularly important role as shown by the principal component analysis as it is related to the Hafnian of all possible subgraphs of $G$ with their first 5 vertices removed. 
We also show that the kernel is sample efficient, a major issue for applications of GBS, and has a time complexity that is comparable with the classical strategies. 

The fact that a GBS kernel using only binary detectors produces such accuracies suggests that technologically more feasible---binary detectors such as SPADs do not operate at cryogenic temperatures such as superconducting PNR ones---GBS devices could have useful applications for machine learning with graph-structured data. We believe that GBS with PNR detectors should also be explored more for this application with particular attention given to coarse-graining strategies that both reduce the sample complexity as well as provide features that capture useful information about the graphs.

A number of questions remain open for investigation such as how vertex and edge labels can be encoded into the GBS device. Also as stated earlier it is known that the existence of a polynomial-time classical algorithm for exact sampling from the output probability distribution of a boson sampling or Gaussian boson sampling device would imply the collapse of the polynomial hierarchy to the third level and thus the existence of such an algorithm is believed to be very unlikely \footnote{Although this has been proven rigorously for the exact sampling case \cite{Aaronson2010, kruse2019detailed} the proof pertaining to the approximate sampling case rests on the assumption of two conjectures known as the Permanent-of-Gaussians Conjecture and the Permanent Anti-Concentration Conjecture which are as of now still unproven.}. This result can also be extended to GBS with binary detectors \cite{quesada2018gaussian}. However it is not known, although some work has been done in this area \cite{bradler2019duality}, if such arguments exist for algorithms that sample from coarse-grained versions of these probability distributions such as those defined in \cite{schuld2020measuring} or our work.
{
It is important to know if such arguments exist as they would imply these quantum kernels are also likely hard to simulate classically. 
}

\begin{acknowledgments}
We thank Maria Schuld, Kamil Br\'adler, Scott Aaronson, Ignacio Cirac, Miller Eaton, Nicolás Quesada, Andrew Blance, and Sefonias Maereg for useful advice and discussions. We  thank Research Computing at the University of Virginia for providing access to, and support with, the Rivanna computing cluster. This work was supported by NSF grant PHY-2112867.
\end{acknowledgments}

\appendix

\section{Reminders about standard GBS\label{app0}}

\subsection{GBS with PNR detectors}
There has been substantial work done already on the connection between graph theory and Gaussian boson sampling with PNR detectors \cite{bradler2018gaussian, bradler2021graph, arrazola2018using}. Here we present the important concepts. Any undirected graph $G$ with no self-loops and $|V|=M$ vertices can be encoded into a $M$-mode GBS setup consisting of a set of $M$ squeezers followed by an interferometer of beamsplitters according to its adjacency matrix $A$. Once the graph is encoded into the GBS device the probability of detecting a specific detection event $\textbf{n} = (n_1,...,n_M)$ is equal to 
\begin{equation}
    p(\textbf{n}) =  \frac{1}{\sqrt{\textrm{det}(Q)}}\frac{\textrm{Haf}(\tilde{A}_{\textbf{n}})}{\textbf{n}!}  =\frac{1}{\sqrt{\textrm{det}(Q)}}\frac{\textrm{Haf}^2(A_{\textbf{n}})}{\textbf{n}!} 
\end{equation}
with
\begin{equation}\label{eq:Q}
    Q = (\mathbb{I}_{2M} - X\tilde{A})^{-1}, \quad
    X =  \begin{bmatrix}
    0 & \mathbb{I}\\
   \mathbb{I} & 0 
    \end{bmatrix},
\end{equation}
$\textbf{n}! = n_1!\times...\times n_M!$, $\tilde{A}=(A \oplus A)$ and Haf() denoting the Hafnian of a $2M \times 2M$ matrix. The Hafnian is a matrix function defined mathematically as 
\begin{equation}
    \textrm{Haf}(A) = \sum_{\pi \in S_M} \prod_{(u,v) \in \pi} A_{u,v},
\end{equation}
where $S_M$ is the partition of the set $\{1,2,...,2M\}$ into unordered disjoint pairs. For example if $M=2$ then $S_M=( \{(1,2), (3,4)\}, \{(1,4),(2,3) \}, \{ (1,3), (2,4)\} )$. If $A$ is the adjacency matrix of an unweighted graph then the Hafnian is equal to the number of perfect matchings of the vertices of the graph. A perfect matching is a partition of the vertex set of a graph into pairs such that each vertex is connected to exactly one edge from the edge set. All perfect matchings of a complete 4-vertex graph are shown in Fig.\ref{fig:haf}.  
\begin{figure}[htb]
\centering
\begin{subfigure}{.2\textwidth}
  \centering
  \includegraphics[scale=0.5]{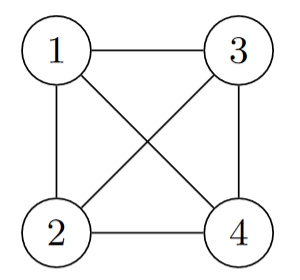}
  \caption{Complete graph of 4 vertices}
  \label{fig:sub3}
\end{subfigure}%
\begin{subfigure}{.3\textwidth}
  \centering
    \includegraphics[scale=0.5]{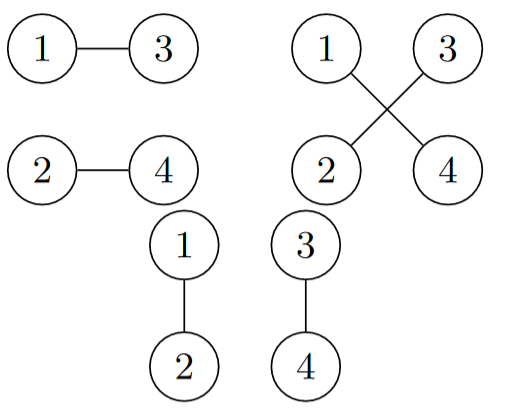}       
  \caption{The three perfect matchings of the complete 4-vertex graph}
  \label{fig:sub4}
\end{subfigure}
\caption{The complete graph of 4 vertices and its corresponding perfect matching}
\label{fig:haf}
\end{figure}

$A_{\textbf{n}}$ is the $n \times n$ submatrix of $A$ induced according to the photon detection event $\textbf{n}$. $A_{\textbf{n}}$ is obtained by repeating the $i$th row and column according to the measurement pattern $\textbf{n}$. If $n_i=0$ then the $i$th row and column are deleted from $A$ but if $n_i>0$ then the $i$th row and column are repeated $n_i$ times. 
For example the probability of detecting the event where each mode has exactly one photon $\textbf{n} = (1, 1, ..., 1)$ would be proportional to the Hafnian of the original matrix $A$ since $A_{\textbf{n}}=A$. 
What this means in terms of the graph is that vertex $i$ and all its edges are either deleted if $n_i=0$ or duplicated $n_i$ times if $n_i > 0$. Therefore the probability of a detection event $\textbf{n}$ is proportional to the squared Hafnian of the subgraph $G_{\textbf{n}}$ corresponding to the induced adjacency matrix $A_{\textbf{n}}$. Examples of different detection events and their corresponding induced subgraphs are shown in Fig.\ref{fig:induced}.
\begin{figure}[htb]
\centering
\includegraphics[scale=0.35]{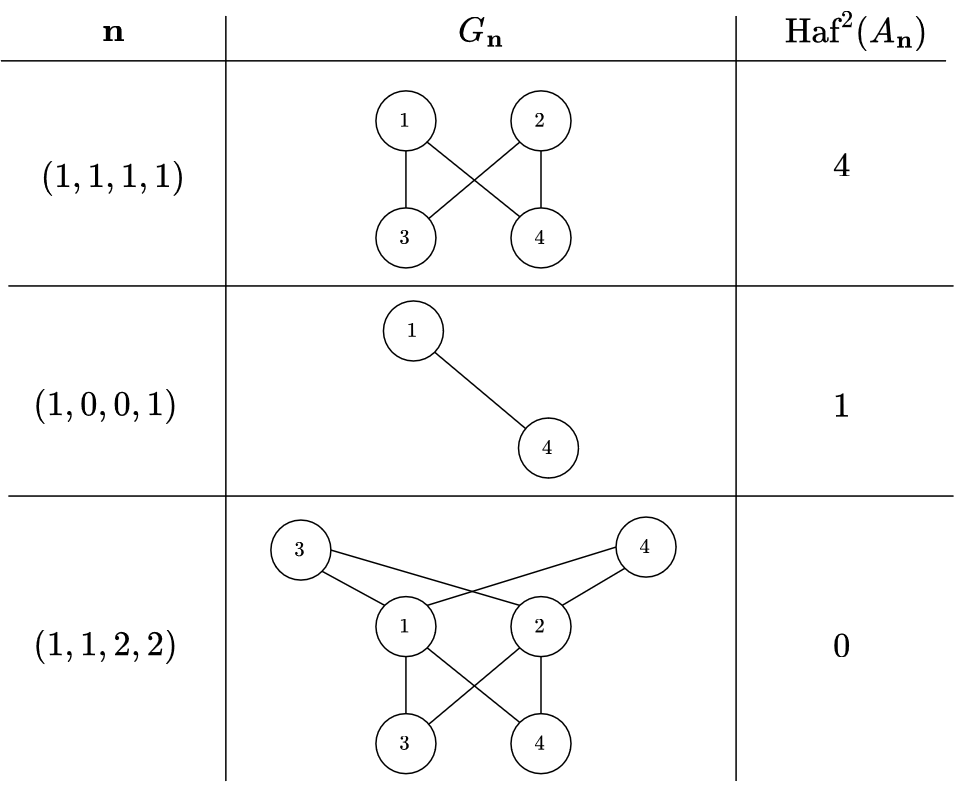}
\caption{
{
Different photon detection events $\textbf{n}$ and the corresponding subgraphs $G_{\textbf{n}}$ they induce  and the value of the squared Hafnians of those subgraphs. The probability of the detection event where each detector detects one photon corresponds to the Hafnian of the graph encoded into the GBS. We can see in the third graph from the top when a detector detects 2 photons the corresponding vertices and their respective edges are duplicated.}}
\label{fig:induced}
\end{figure}

These induced subgraphs are of even size since the number of photons detected is always even due to the fact that the inputs are squeezed states. However when displacement is applied to the modes of the GBS the probability of detecting an odd number of photons is in general not zero anymore and the probability of individual detection events is characterized by the loop Hafnian lHaf() as opposed to the Hafnian \cite{quesada2019franck,bulmer2022threshold}. We do not apply displacement for the numerical experiments done in this paper.

\subsection{Encoding a graph into a GBS device}

To map a graph to a feature vector we must first program the GBS device, by setting the squeezing parameters and beamsplitter angles of the device, according to the adjacency matrix $A$ of the graph. Any adjacency matrix $A \in \mathbb{R}^{M \times M}$ of an undirected graph of $M$ vertices can be mapped to a symmetric, positive definite $2M \times 2M$ covariance matrix $\boldsymbol{\Sigma}$ of a pure Gaussian state of $M$ modes via the following procedure. First a doubled adjacency matrix $\tilde{A}$ is constructed,
\begin{equation}
        \tilde{A} = c \begin{bmatrix}
    A & 0\\
    0 & A
    \end{bmatrix} = c(A \oplus A), 
\end{equation} 
where $c$ is a rescaling constant chosen such that $0 < c < 1/\lambda_{\textrm{max}}$ where $\lambda_{\textrm{max}}$ is the maximum singular value of $A$ \cite{schuld2020measuring}. We use $\tilde{A}$ as, unlike $A$, it is guaranteed to map to a covariance matrix of a pure Gaussian state which is easier to prepare than a mixed one \cite{bradler2018gaussian}. This also has the advantage of allowing us to utilize the identity $\textrm{Haf}(A \oplus A) = \textrm{Haf}^2(A)$ to relate $\Tilde{A}$ to $A$. To map $\tilde{A}$ to a covariance matrix $\boldsymbol{\Sigma}$ we use the following matrix equations
\begin{equation}
    \boldsymbol{\Sigma} = Q - \mathbb{I}_{2M}/2, \textrm{ with } Q = (\mathbb{I}_{2M} - X\tilde{A})^{-1}, \quad
    X =  \begin{bmatrix}
    0 & \mathbb{I}\\
   \mathbb{I} & 0 
    \end{bmatrix}.
\end{equation}
To program the GBS device to sample from the probability distribution corresponding to the covariance matrix $\boldsymbol{\Sigma}$ of the pure Gaussian state we need the unitary matrix $U$ that characterizes the interferometer of the device as well as the squeezing parameters $r_1,...,r_M$ of each of the $M$ squeezers. We can obtain these values by taking the Takagi decomposition of $A$ which is of the form
\begin{equation}
    A = U \textrm{diag}(\lambda_1,...,\lambda_M)U^T.
\end{equation}
The squeezing parameters are determined by the singular values $\lambda_1,...,\lambda_M$  and $c$ via the relationship $r_i=\tanh^{-1}(c\lambda_i)$. 
The singular values and $c$ also uniquely determine the mean photon number $\bar{n}$ of the device according to
\begin{equation}
     \bar{n} = \sum_{i=1}^{M} \frac{(c\lambda_i)^2}{1-(c\lambda_i)^2}=\sum_{i=1}^{M}\sinh^2(r_i).
\end{equation}
The rescaling constant $c$ can be used to adjust $\bar{n}$ as multiplying $A$ by $c$ scales it's singular values without changing the structure of the graph other than scaling all edge weights by $c$.
The matrix $U$ can be decomposed to give the parameters of the beamsplitter gates of the interferometer \cite{clements2016optimal}.

The GBS device, if using PNR detectors, now samples from the probability distribution
\begin{equation}
    p(\textbf{n}) = \frac{1}{\sqrt{\textrm{det}(Q)}}\frac{\textrm{Haf}(\tilde{A}_{\textbf{n}})}{\textbf{n}!} = \frac{1}{\sqrt{\textrm{det}(Q)}}\frac{\textrm{Haf}^2(A_{\textbf{n}})}{\textbf{n}!}.
\end{equation}

\section{Super Exponential Growth of GBS Detection Events for $M \in \mathcal{O}(n^2)$ \label{app1}}
\begin{lemma}
$\frac{(n+M-1)!}{n!(M-1)!}\in \omega(\floor{\sqrt{M}}^{\floor{{\sqrt{M}}}})$ \textrm{for} $n = \floor{\sqrt{M}}$ 
\end{lemma}
\begin{proof}
    \begin{align*}
    \frac{(n+M-1)!}{n!(M-1)!} &\xrightarrow{n = \floor{\sqrt{M}}}  \frac{(\floor{\sqrt{M}}+M-1)!}{(\floor{\sqrt{M}})!(M-1)!} \\ 
    \frac{(\floor{\sqrt{M}}+M-1)!}{(\floor{\sqrt{M}})!(M-1)!}&= \frac{[\prod_{i=1}^{\floor{\sqrt{M}}} (M-1+i)](M-1)!}{[\prod_{i=1}^{\floor{\sqrt{M}}}i] (M-1)!} \\ 
    &= \frac{[\prod_{i=1}^{\floor{\sqrt{M}}}  (M-1+i)]}{[\prod_{i=1}^{\floor{\sqrt{M}}}i]} \\ 
    &
    = \prod_{i=1}^{\floor{\sqrt{M}}}[\frac{M-1}{i} + 1] \\
    &> \prod_{i=1}^{\floor{\sqrt{M}}}[\frac{M-1}{\floor{\sqrt{M}}} + 1] \\
    &= (\frac{M-1}{\floor{\sqrt{M}}} + 1)^{\floor{\sqrt{M}}} \\
    &= (\frac{M}{\floor{\sqrt{M}}} + 1 - \frac{1}{\floor{\sqrt{M}}})^{\floor{\sqrt{M}}} \\
    &\geq {\floor{\sqrt{M}}}^{\floor{\sqrt{M}}}
    \end{align*}
   {
   Therefore $\frac{(n+M-1)!}{n!(M-1)!} \in \omega(\floor{\sqrt{M}}^{\floor{{\sqrt{M}}}}) $ for $n = \floor{\sqrt{M}}$.}
\end{proof}
\section{Induction Proof for ${n \choose k} \in \Theta(n^k)$ \label{app2}}
\begin{lemma}
${n \choose k} \in \Theta(n^k)$
\end{lemma}
\begin{proof}
    \textrm{Base Case: $k=2$\\} 
    \begin{align*}
    {n \choose 2} 
    &= \frac{n(n-1)}{2!} 
    \end{align*}
    \begin{align*}
    \lim_{n\to\infty} \frac{\frac{n(n-1)}{2!}}{n^2} = \frac{1}{2!} \\
    0 < \frac{1}{2!} < \infty \\ 	\therefore {n \choose 2} \in \Theta(n^2)
    \end{align*}
    {
    \textrm{Assume result holds up to $k=\ell$} 
    \begin{align*}
    {n \choose \ell} 
    &= \frac{n(n-1)(n-2) \cdots (n-\ell+1)}{\ell!} \in \Theta(n^\ell)
    \end{align*}
    Inductive Step: $k=\ell+1$ 
    \begin{align*}
    {n \choose \ell+1} 
    &= \frac{n(n-1)(n-2) \cdots (n-\ell)}{(\ell+1)!}  
    \end{align*}
    \begin{align*}
    &\lim_{n\to\infty} \frac{\frac{n(n-1)(n-2) \cdots (n-\ell)}{(\ell+1)!}}{n^{\ell+1}} \\
    =&\lim_{n\to\infty} \frac{\frac{n(n-1)(n-2) \cdots (n-\ell+1)}{\ell!}}{n^\ell}\frac{\frac{(n-\ell)}{\ell+1}}{n} \\
    =&\lim_{n\to\infty} \frac{\frac{n(n-1)(n-2) \cdots (n-\ell+1)}{\ell!}}{n^\ell}\lim_{n\to\infty}\frac{\frac{(n-\ell)}{\ell+1}}{n} \\
    \equiv &\frac{1}{\ell!}\frac{1}{\ell+1}\\
    =&\frac{1}{(\ell+1)!} \\
    & 0 < \frac{1}{(\ell+1)!} < \infty \\ 
    &\therefore {n \choose \ell+1} \in \Theta(n^{\ell+1})
    \end{align*}
    }
\end{proof}


\bibliography{GBSbin,Pfister}

\end{document}